\documentclass{amsart}
\usepackage{graphicx}
\usepackage{amsthm,amsmath,verbatim,amssymb}
\usepackage[toc,page]{appendix}
\usepackage{color}
\usepackage{arydshln}
\newtheorem{lem}{Lemma}

\newtheorem{thm}{Theorem}

\DeclareMathOperator{\Tr}{Tr}

\newcommand{\p}{\psi}
\newcommand{\la}{\lambda}\renewcommand{\l}{\lambda}
\newcommand{\e}{\mathbb{E}}

\newtheorem{rem}{Remark}
\begin{document}
\title[Spectrum]{Spectrum of SYK model }

\author[Feng, Tian and Wei]{Renjie Feng, Gang Tian, Dongyi Wei}
\address{Beijing International Center for Mathematical Research, Peking University, Beijing, China, 100871.}

\email{renjie@math.pku.edu.cn}
\email{gtian@math.pku.edu.cn}
\email{jnwdyi@pku.edu.cn}

\date{\today}
   \maketitle
    \begin{abstract}
This is the first part of a series of papers on the spectrum of the SYK model, which is a simple model of the black hole in physics literature. In this paper, we will give a  rigorous proof  of the almost sure convergence of the global density of the eigenvalues.  We also discuss the largest eigenvalue of the SYK model.

\end{abstract}

\section{Introduction}
In the 1990s, to study the new quantum phase which is called the quantum spin glass and non-Fermi liquid, Sachdev and Ye \cite{SY} proposed a  model that describes random interacting spins with infinite range. Based on this early work, the Sachdev-Ye-Kitaev (SYK) model, which is proposed by Kiteav \cite{K}, describes $n$ random interacting Majorana modes on a quantum dot, and suggest  the possible holographic description of the SYK model  after taking large $n$ limit.

The SYK model is a random matrix model where the main interest lies in the global and local behaviors of its eigenvalues in mathematics. Actually
physicists have many results regarding the spectrum of SYK, either theoretically or numerically \cite{black, GV, GV2, GV3,GV4, LNZ, MS,bulk spectrum Polchinski}. 
However, mathematical aspects of the SYK model were less studied.  In this paper, we will prove that the normalized  empirical measure of eigenvalues converges to some limiting measure with probability 1 (or almost surely), which can be viewed as a \emph{law of large numbers} in probability theory. In the end, we will discuss the asymptotic behavior of the largest eigenvalue. 

In our subsequent papers \cite{FTW2, FTW3}, we further derive two theorems about the spectrum  of the SYK model. The results are totally unknown in physics, but they are indeed the most fundamental and important theorems considered in random matrix theory. To be more precise, in \cite{FTW2}, we  prove the \emph{central limit theorem} for the linear statistic of  of eigenvalues  as $n\to\infty$ and derive an explicit formula for its variance. These results imply some useful information about the (global) 2-point correlation of the eigenvalues.  In \cite{FTW3},  for the special case of the Gaussian SYK model,  we  will derive a \emph{large deviation principle} for the normalized empirical measure of eigenvalues for $q_n=2$ (in which case it's a totally solvable system and  physicists do not care it too much, but it does have its own interest in random matrix theory) and a \emph{concentration of measure theorem} for general $q_n\geq 3$. 

\subsection{SYK model}
Throughout the article, let $n$ be an even integer.  Let's first assume $q_n$ is even and $ 2\leq {q_n}< n$.
Then we   consider the following Hamiltonian \begin{equation}\label{sykm}H=i^{q_n/2}\frac{1}{\sqrt{{{n} \choose {q_n}}}}\sum_{1\leq i_1<i_2< \cdots < i_{q_n} \leq n} J_{i_1i_2\cdots i_{q_n}}{\p}_{i_1}{\p}_{i_2}\cdots {\p}_{ i_{q_n}},\end{equation}
where the real random variables $J_{i_1i_2\cdots  i_{q_n}} $ are independent identically distributed (i.i.d.),  nondegenerated  and
$$\e J_{i_1i_2\cdots  i_{q_n}}=0,\,\,\,\, \e J_{i_1i_2\cdots  i_{q_n}}^2=1, $$
 and the $k$-th moment  of $|J_{i_1i_2\cdots  i_{q_n}}| $ is uniformly bounded for any fixed $k$;
$\p_j$ are Majorana fermions which obey the algebra
\begin{equation}\label{anticom}\{\p_i,\p_j\}:=\p_i\p_j+\p_j\p_i=2\delta_{ij}  .\end{equation}
 In fact, by the representation of the Clifford algebra, each $\p_i$ is a $2^{n/2}\times 2^{n/2}$ Hermitian matrix  generated by Pauli matrices  iteratively \cite{LM}. Let's denote $L_n=2^{n/2}$.
Note that we do not assume $J_{i_1i_2\cdots  i_{q_n}} $ to be Gaussian, the results are true for more general random variables.

Actually, the original SYK model is
\begin{equation}\label{sykm2}H_{SYK}=i^{q_n/2}\frac{1}{\sqrt{\frac {N^{q_n-1}} {(q_n-1)!J^2}}}\sum_{1\leq i_1<i_2< \cdots <i_{q_n} \leq n} J_{i_1i_2\cdots i_{q_n}}{\p}_{i_1}{\p}_{i_2}\cdots {\p}_{ i_{q_n}},\end{equation}
where the random variables $ J_{i_1i_2\cdots i_{q_n}}$ satisfy the same assumptions as in \eqref{sykm} and $J$ is some constant. There is no essential difference between \eqref{sykm} and \eqref{sykm2} other than a normalizing constant. In this paper, we will use the random matrix $ H$ instead of $H_{SYK}$.

\subsection{Global  density}
Let's first state the results on the global  density of the eigenvalues. Let  $\la_{i}, 1\leq i \leq L_n$ be the eigenvalues of $H$. One can check that $H$ is Hermitian by the anticommutative relation \eqref{anticom}, and thus $\la_{i}$ are real numbers.
The normalized empirical measure of the eigenvalues is defined as
	\begin{equation}\label{emph}
  \rho_n(\lambda)=\frac 1{L_n}\sum_{i} \delta_{\la_{i}}(\la).
	\end{equation}
 One of the main results in this paper is that $\rho_n$ will converge to a probability measure $\rho_\infty$ almost surely  in the sense of distribution, and there is a phase transition in the density of the states depending on the limit of the quotient $ q_n^2/n$,

\begin{thm}\label{main1}  Let $q_n,n$ be even and $2\leq q_n\leq n/2$. Let  the random variables $J_{i_1\cdots  i_{q_n}}$  be i.i.d. and nondegenerated with expectation 0 and variance $1$ and the $k$-th moment of $|J_{i_1\cdots  i_{q_n}}|$  is uniformly bounded for any $k$.    Then the normalized empirical measure $\rho_n$ of eigenvalues of the random matrix $H$ defined in \eqref{sykm} will converge to $\rho_\infty$ almost surely in the sense of distribution, where the probability measure $\rho_\infty$ is given explicitly as follows,   \\
1. When $ q_n^2 /  n\to 0$, then $\rho_\infty$ is the standard Gaussian distribution. \\
2. When  $ q_n^2 /  n\to a$, then $\rho_\infty$ has compact support with the density function
$$p_a(x) =
 \begin{cases}
    f(x| e^{-2a})   &\text{if $x\in [-\frac 2{\sqrt{1-e^{-2a}}},\frac 2{\sqrt{1-e^{-2a}}}]$,} \\
    0 &  \text{else,}
 \end{cases} $$
where the function  \begin{equation}\label{fs}f(x|y)=\frac{\sqrt{1-y}}{\pi\sqrt{1-(1-y)x^2/4}}
\prod_{k=0}^\infty\left[\frac{1-y^{2k+2}}{1-y^{2k+1}}(1-\frac{x^2(1-y)y^k}{(1+y^k)^2})\right]. \end{equation}
3. When $q_n^2/n\to \infty$,  the limiting density satisfies the semicircle law
$$\rho_{\infty}(x)=\frac{1}{2\pi}\sqrt{4-x^2} \mbox{\Large$\chi$}_{[-2,2]}.$$
\end{thm}
As a remark, as $a\to +\infty$, we have the limit
$$\lim_{a\to +\infty}p_a(x)=p_{\infty}(x)=\frac{1}{2\pi}\sqrt{4-x^2}\mbox{\Large$\chi$}_{[-2,2]}.$$
  As $a\to0$, $p_0(x)$ is proved to be the standard Gaussian distribution (see section 2 in \cite{ISV}), i.e.,
$$\lim_{a\to 0}p_a(x)=p_0(x)=\frac1{\sqrt{2\pi}}e^{-\frac {x^2}2}.$$
The above two  limits indicate that case 2 in Theorem \ref{main1} yields a phase transition between case 1 and case 3.

We will prove Theorem \ref{main1} by the moment method: we will first prove that the expectation of $\mathbb \rho_n$ tends to $\rho_\infty$
 in \S\ref{expected}, the claim of the almost sure convergence follows the estimate of the variance  in \S\ref{almost}.

\begin{rem} There are several new features about our results compared with these in physics. 
Physicists only care the Gaussian SYK model, i.e., when all random variables  $J_{i_1\cdots  i_{q_n}}$ are i.i.d standard Gaussian. They have  proved cases 1 and 2 in  Theorem \ref{main1}   for the Gaussian SYK model (cf. \cite{black, GV, GV2}); but for case 3,
there is only a heuristic proof by some physics method making use of the Grassmann integral \cite{LNZ}. In our paper, we will derive the limit of the global density of eigenvalues for more general random variables  rigorously, especially for the case 3. Furthermore, we prove that the convergence  is with probability 1 (or almost surely).

\end{rem}
As another remark, a related model is the quantum $q$-spin glass model considered in \cite{So, KLW, KLW2}. For $q\geq 1$, the Hamiltonian of a quantum $q$-spin glass is
$$H^q_{  n}=3^{-q/2}{n\choose q}^{-1/2}\sum_{1\leq i_1<\cdots<i_{q}\leq n}\sum_{a_1,\cdots, a_q=1}^3\alpha_{a_1,\cdots, a_q,(i_1,\cdots, i_q)} \sigma_{i_1}^{a_1}\cdots \sigma_{i_q}^{a_q},$$
where the coefficients $\alpha_{a_1,\cdots, a_q,(i_1,\cdots, i_q)}$
are i.i.d. random variables with expectation $0$ and variance $1$ and
\begin{equation}\label{sigma}\sigma_i^a=I_2^{\otimes (i-1)}\otimes \sigma^a\otimes I_2^{\otimes (n-i)}\end{equation}
where $\sigma^a$, $a=1,2,3$ are three Pauli matrices and $I_2$ is the $2\times 2$ identity matrix. Thus $\sigma_i^a$ is a $2^n\times 2^n$ matrix.
Erd\H{o}s-Schr\"oder proved in \cite{So} that the limiting density of $H^q_{  n}$ has a  phase transition  in the regimes of $q^2\ll n, q^2/n\to a$ and $q^2\gg n$ and the similar results as Theorem \ref{main1} can be derived. 

\begin{rem} Our proof of cases 1 and 2 follows closely with that of Erd\H{o}s-Schr\"oder in \cite{So}, but  the proof of case 3 is quite different. This is because the matrices \eqref{sigma} are explicitly given and one can make use of this explicit construction  to derive the result. But in our case,  the essential difficulty is that we  can only apply the anticommutative relation \eqref{anticom}, so we have to take a totally different approach (see Lemma \ref{lemma5} below).
\end{rem}


Our proof can be applied to the case when $q_n$ is odd where \begin{equation}\label{sykm1}H=i^{(q_n-1)/2}\frac{1}{\sqrt{{{n} \choose {q_n}}}}\sum_{1\leq i_1<i_2< \cdots <i_{q_n} \leq n} J_{i_1i_2\cdots i_{q_n}}{\p}_{i_1}{\p}_{i_2}\cdots {\p}_{i_{q_n}}.\end{equation} Following the   proof of Theorem \ref{main1},  in \S \ref{theorem2proof}, we will prove
 \begin{thm}\label{main2}Let $n$ be even and $1\leq q_n\leq n/2$ be odd. The random variables $J_{i_1i_2\cdots i_{q_n}}$ satisfy the same assumptions as in Theorem \ref{main1}. Then   the normalized empirical measure  $\rho_n$  of the eigenvalues of   random matrices of \eqref{sykm1}  will converge to $\rho_\infty$ almost surely in the sense of distribution where \\
1. When $q_n^2/n\to 0$, then $\rho_\infty=\frac{1}{2}(\delta_{1}+\delta_{-1})$. \\
2. When $q_n^2/n\to a$, then  the density function of $\rho_\infty$ is $$p_a(x) =
 \begin{cases}
    f(x| -e^{-2a})   &\text{if $x\in [-\frac 2{\sqrt{1+e^{-2a}}},\frac 2{\sqrt{1+e^{-2a}}}]$,} \\
    0 &  \text{else,}
 \end{cases} $$
where the function  $ f(x|y)$ is given by \eqref{fs}.\\
3. When $q_n^2/n\to \infty$,  the limiting density $\rho_\infty$ is still the semicircle law. \end{thm}
Moreover, following  \cite{ISV}, one can also prove that there is a phase transition for this odd case. 


 Given $H$ with $q_n$ product of Majorana fermions,   $\psi_1\cdots\psi_n H$ will
yield another Hamiltonian with $n-q_n$ product of fermions, i.e.,
$$\psi_1\cdots\psi_n H=i^{q_n/2}\frac{1}{\sqrt{{{n} \choose {q_n}}}}\sum_{ } \pm J_{i_1i_2\cdots i_{q_n}}\p_1\cdots \hat{\p}_{i_1}\cdots \hat {\p}_{i_{q_n}}\cdots \psi_n,$$
where $\hat {\p} $ denotes the omitting of the fermion ${\p} $.
This implies that  there is a symmetry between the systems of $q_n$ and $n-q_n$. Therefore, Theorems \ref{main1} and \ref{main2} can be extended to $q_n\geq  n/ 2$ immediately.
For example,  given $H$ with even $q_n\geq  n/ 2$,  let $\widetilde{H}=i^{n/2-q_n}\psi_1\cdots\psi_n H$, then $\widetilde{H} $ corresponds to the case $n-q_n\leq n/2$. We first have
 $H^2=\widetilde{H}^2$ and thus $ L_n^{-1}\Tr H^k=L_n^{-1}\Tr \widetilde{H}^k$ for $k$ even.
Moreover,  Lemma \ref{lemma1} also implies $ \e [L_n^{-1}\Tr {H}^k]\to 0$ for $k$ odd. Therefore, the moment method implies that Theorem \ref{main1} holds for $H$ with $q_n^2/n\to a\in[0,+\infty] $ replaced by $(n-q_n)^2/n\to a\in[0,+\infty]$.

\subsection{Largest eigenvalues}One of the central studies in random matrix theory is about the largest eigenvalue of random matrices. There are many classical results regarding the largest eigenvalue for   general Wigner matrices, such as Bai-Yin's result on the almost sure convergence of the largest eigenvalue and the Tracy-Widom law about the rescaling of the largest eigenvalue (see \cite{AGZ, T}). 

 	In \S \ref{largg}, we will prove two easy results regarding  the largest eigenvalue of the SYK model.  When $q_n=2$, the system is totally solvable and it seems that physicists know how to derive the largest eigenvalue (cf. page. 4 in \cite{MS}). Actually,  if  $q_n=2$, we may consider $J_{ij}$ as a random antisymmetric matrix, i.e., $B:=(J_{ij})_{1\leq i,j\leq n},\ J_{ji}=-J_{ij} $ is a real random antisymmetric matrix. We assume the eigenvalues of $B$ are $\pm i\mu_j$ where  $\mu_j\geq 0$ for $ 1\leq j\leq n/2$. Then it's well-known that all eigenvalues of the SYK model can be expressed in term of $u_j,\,j=1,\cdots, n/2$ by identity \eqref{expressions}, we also refer to \cite{black, GV3}. Actually, the largest eigenvalue is
 $$\lambda_{max}= {{n} \choose {2}}^{-\frac{1}{2}}\sum\limits_{j=1}^{n/2} \mu_{j}.$$
 Then by the classical results on the distribution of eigenvalues of random antisymmetric matrices, we  easily have \begin{thm}\label{2case}For $q_n=2$, the largest eigenvalue of $H$ satisfies
 $$\lim_{n\to+\infty}\frac{\lambda_{max}}{\sqrt{n}}=\frac{4\sqrt{2}}{3\pi}, \ a.s.$$
 \end{thm}
For general $q_n$,  we do not know how to derive the asymptotic limit of the largest eigenvalue of the SYK model, although there are many numerical results \cite{GV, GV2}. If we consider the Gaussian SYK model, we only have the  rough upper bound, 
\begin{thm}\label{gausscase}
For the Gaussian SYK model where $J_{i_1\cdots i_{q_n}}$ are i.i.d. standard Gaussian random variables, let $q_n\geq 4$ be even,
then we have  \begin{equation}\e \la_{max}\leq \sqrt{n\ln 2}.\end{equation}
\end{thm}

 
 \subsection{Further discussions}
 There are still many  problems and essential difficulties about the SYK model. The moment method is an effective way to study the global properties of the SYK model, but it will fail  when studying the local behaviors. For example,  one may  use the moment method to prove the almost sure convergence of the largest eigenvalue for general Wigner matrices \cite{AGZ}. But for the SYK model, the moment method does not work, this is  basically because the SYK model is a sparse matrix,  where the size of the matrix is of exponential growth but the number of nonzero elements is of polynomial growth. Therefore,  it seems impossible to estimate the trace $\Tr H^{k_n}$ where the power  $k_n$ is some function  growing with $n$ as in the classical moment method (see Sinai-Soshnikov's proof  \cite{Ss} and Bai-Yin's proof in \cite{T}). 

Another difficulty about the SYK model  is that the resolvent method which is found to be one of the most powerful methods in random matrix theory fails. Briefly, the idea of the resolvent method is that the  Stieltjes transforms of $N\times N$ Wigner matrices can be expressed in term of some 
$(N-1)\times (N-1)$ matrices by the Cramer's rule, and one may estimate these matrices to derive recursively an equation about the  Stieltjes transforms  as $N\to\infty$ (see \cite{AGZ, T}).  But so far, we have not found the right way to apply the resolvent method to the SYK model. 

To conclude, let's discuss some further problems for the SYK model. 

1. Regarding the largest eigenvalue,  let's assume $q_n$ is even, when  $q_n^2/n\to 0$, does  $\lim_{n\to \infty} \e \lambda_{\max}/\sqrt n$ exist? If it exists, how does it depend on $q_n$?  When $q_n^2/n\to a\in (0,\infty]$,  does $ \lambda_{\max}$ converge to  $\frac 2{\sqrt{1-e^{-2a}}}$ almost surely? 
 
Let's discuss the first question. Let's assume that the random variables are i.i.d. standard Gaussian. For every fixed $q_n\geq 4,\ q_n$ even, we know that $\e \la_{max}/\sqrt{n} $ is bounded from above by Theorem \ref{gausscase}, but it seems very difficult to prove its convergence. 	
For any $ \beta>0$, let's define the partition function
\begin{equation}\label{parti}Z(\beta)=\Tr e^{-\beta H}. \end{equation} Then we easily have  $$\la_{max}\leq (\sqrt{n}\beta)^{-1}\ln Z(-\sqrt{n}\beta)\leq \la_{max}+(2\beta)^{-1} \sqrt{n}\ln 2.$$  Therefore, one possible approach to prove the convergence of $\e \la_{max}/\sqrt{n} $ is to prove the convergence of $n^{-1}\e\ln Z(-\sqrt{n }\beta),$ then let $\beta\to +\infty.$ It seems that one may apply the  idea of the Parisi formula for the   Sherrington-Kirkpatrick (SK) model to prove this \cite{ta2}.
But compared with the SK model, the main difference of the SYK model is that the product of Majorana fermions may not commute with each other, this causes many difficulties.

 2. 
There are some results in physics regarding the rescaling of the distribution around the ground state of the SYK model, especially for the case $q_n=4$ \cite{black, GV, GV2, SY2}. To consist with these results, let's consider the original SYK model in \eqref{sykm2} for $q_n=4$, $$H_{SYK}=\frac{1}{\sqrt{\frac {n^{3}} {6J^2}}}\sum_{1\leq i_1<i_2< i_3 <i_4 \leq n} J_{i_1i_2i_3 i_{4}}{\p}_{i_1}{\p}_{i_2}{\p}_{i_3}{\p}_{ i_{4}},$$
where $J_{i_1i_2i_3 i_{4}}$ are i.i.d. Gaussian distribution.  Let $$\tilde\rho_n(\lambda):=\sum_{i} \delta_{\la_{i}}(\la)$$ be the empirical measure of eigenvalues. One of the main results in physics (see equation B. 15 in \cite{black} and also the numerical results in \cite{GV2})  is the following rescaling limit for $\lambda- a_0 \sim  1/n$ as $n$ large enough,\begin{equation}
	\label{timm}
 \e \tilde \rho_n(\lambda)\propto e^{nS_0}\sinh \sqrt {nC(\lambda-a_0)} \end{equation}
where $a_0=\alpha n$ for some $\alpha<0$ and the constants $S_0$ and $C$ are independent of $n$.
derived in \cite{MS} by the physics method of the path integral, there are also some numerical results in \cite{GV}.  The constant $S_0$ is universal, but the constant $C$ is non-universal, depending on the distribution of the $J_{i_1i_2i_3i_4}$, and is not known exactly.

 The method in physics to derive such rescaling limit is by the double-scaled limit (see Appendix B in \cite{black}), where the idea   is based on the assumption that the ground state follows the same phase transition as the global density in Theorem \ref{main1}.  Such assumption seems quite reasonable, but we do not have any rigorous proof, and hence the proof of \eqref{timm}.

3. In this paper, we can derive the limiting density $\rho_\infty$ of the global density, but it seems very reasonable to believe that the global density has a full expansion with lower order terms. This problem is considered in \cite{GV2, GV3, GV4} and the first two lower order terms are computed numerically,  thus it's also very meaning to derive the lower order terms mathematically.

\bigskip
\textbf{Acknowledgement:} We thank Subir Sachdev for bringing our attention to the SYK model and sending us a short note on the SYK model, which is
very helpful for us to
start this project. We also thank Gerard Ben Arous, Zhi-Dong Bai, Peter J. Forrester, Dang-Zheng Liu, Douglas Stanford for many helpful discussions.

\section{Expected density}\label{expected}
In this section, we will  prove that when $2\leq q_n\leq n/2$ is even, the expectation of the normalized empirical measure $\mathbb \rho_n$ will tend to $\rho_\infty$ as in Theorem \ref{main1} in the sense of distribution.
The proof is based on the moment method. 

\subsection{Notations}\label{notations}Let's introduce some notations.
  For a set $A=\{i_1, i_2,\cdots, i_{m}\} \subseteq\{1,2,\cdots,n\}$, $ 1\leq i_1<i_2< \cdots <i_m \leq n,$ we denote $$\Psi_A:=\psi_{i_1}\cdots \psi_{i_{m}} \,\,\,\mbox{and}\,\,\,\Psi_A:=I\,\,\mbox{if}\,\,A=\emptyset.$$ We denote the set   $$I_n=\{(i_1, i_2,\cdots, i_{q_n}),\,\,1\leq i_1<i_2< \cdots <i_{q_n} \leq n \}.$$Thus the cardinality of $I_n$ is $$|I_n|={n\choose q_n}.$$ For any coordinate $R=(i_1, \cdots, i_{q_n})\in I_n$, we denote $$J_R:=J_{i_1\cdots i_{q_n}}\,\,\,\mbox{and}\,\,\,\,\Psi_R:=\psi_{i_1}\cdots \psi_{i_{q_n}}. $$ Sometimes we identify $R$ with the set $\{i_1, \cdots, i_{q_n}\}.$

Given any set $X$ and any integer $k\geq 1$, we define $P_2(X^k)$
 to be   the tuples $(x_1,\cdots, x_k) \in X^k$
for
which all entries $x_1,\cdots, x_k$ appear exactly twice. If $k$ is odd, then $P_2(X^k)$ is an empty set.

Throughout the article, we denote $c_k$
as some constant  depending only on $k$ and independent of $n$ and $q_n$, but its value may differ from line to line, the same  for $c_{2k}$, $c_k'$  and so forth.

\subsection{Some basic properties on fermions}\label{basicfermion}
Let's  derive some basic properties about the product of Majorana fermions that  will be applied many times in the article.
By relation \eqref{anticom}, we easily have\begin{itemize}

\item If $|A|$ is odd, $i\in A$ or $|A|$ is even, $i\not\in A$, then $\Psi_A\psi_i=\psi_i\Psi_A $.

\item If $|A|$ is even, $i\in A$ or $|A|$ is odd, $i\not\in A$, then $\Psi_A\psi_i=-\psi_i\Psi_A $.

\item $\Psi_A\Psi_B=\pm\Psi_{A\triangle B}, $ here we denote $A\triangle B:=(A\setminus B)\cup (B\setminus A)$.

\end{itemize}
Using these properties, we know that
given a set $A\subseteq\{1,2,\cdots,n\}$, if $|A|$ is odd, then $|A|<n,$ taking $i\in\{1,2,\cdots,n\}\setminus A$, then we have $\Tr\Psi_A=\Tr\psi_i\Psi_A\psi_i^{-1}=\Tr(-\Psi_A); $ if $|A|$ is even and $A\neq \emptyset,$ taking $i\in A$, then we still have $\Tr\Psi_A=\Tr\psi_i\Psi_A\psi_i^{-1}=\Tr(-\Psi_A).$ Therefore,  we show that
\begin{itemize}

\item$\Tr\Psi_A=0 $ and $\Psi_A\neq \pm I $ are always true for $A\neq \emptyset.$\end{itemize} For $ {R_1},\cdots ,{R_k}\subseteq\{1,2,\cdots,n\}$, we have $\Psi_{R_1}\cdots \Psi_{R_k}=\pm\Psi_{R} $ for some $R\subseteq\{1,2,\cdots,n\}$. Thus, if $\Psi_{R_1}\cdots \Psi_{R_k}\neq\pm I$, then $R\neq \emptyset$ and $\Tr \Psi_{R_1}\cdots \Psi_{R_k}= 0 ;$ if $\Psi_{R_1}\cdots \Psi_{R_k}=\pm I$, then $\Tr \Psi_{R_1}\cdots \Psi_{R_k}= \pm L_n .$ In both cases, we always have \begin{equation}\label{tr}
|\frac 1 {L_n} \Tr \Psi_{R_1}\cdots \Psi_{R_k}|\leq 1.
 \end{equation} For $A,B\subseteq\{1,2,\cdots,n\}$, if $\Psi_A=\pm\Psi_B,$ then $\Psi_{A\triangle B}=\pm\Psi_A\Psi_B=\pm\Psi_A\Psi_A=\pm\Psi_{A\triangle A}=\pm\Psi_{\emptyset}=\pm I, $ then we must have $A\triangle B=\emptyset$ and $A=B$, i.e.,
  \begin{itemize}

\item $\Psi_A=\pm\Psi_B$ if and only if $A=B$.
\end{itemize}

\subsection{Expectation}
Now we can turn to prove that the expectation of the normalized empirical measure $\rho_n$ tends to $\rho_\infty $ as in Theorem \ref{main1} when $2\leq q_n\leq n/2$ is even. We will need several lemmas but  their  proofs are postponed to \S\ref{lemmas}.

We first rewrite \begin{equation}\label{mnk}m_{n,k}   = \frac 1{L_n}   \e[ \Tr H^k]= \frac 1{L_n} \frac {i^{q_nk/2}}{{n\choose q_n}^{k/2}} \sum_{{R_1},...,  {R_k}\in I_n} \e[J_{R_1}\cdots J_{R_k}]  \Tr \Psi_{R_1}\cdots \Psi_{R_k}. \end{equation}

As in \cite{So}, we divide the above summation into several  parts and bound each part separately. We first divide the summation of \eqref{mnk} as
\begin{equation}\label{deco}m_{n,k}=\sum_{({R_1},...,  {R_k})\in  P_2(I_n^k)}+\sum_{({R_1},...,  {R_k})\in I^k_n\setminus P_2(I_n^k)}.  \end{equation}
We first have
\begin{equation} \label{exp}  \e[J_{R_1}\cdots J_{R_k}] =1 \,\,\, \mbox{for}\,\,\, ({R_1},...,  {R_k})\in  P_2(I_n^k).
\end{equation}
The following lemma shows that the second summation for $({R_1},...,  {R_k})\in I^k_n\setminus P_2(I_n^k)$ tends to vanish as $ n\to\infty$.

\begin{lem} \label{lemma1}For the summation $({R_1},...,  {R_k})$ that do not appear exactly twice, we have the uniform estimates
 $$\frac 1{{n\choose q_n}^{k/2}}  \sum_{({R_1},...,  {R_k})\in I^k_n\setminus P_2(I_n^k)}| \e[J_{R_1}\cdots J_{R_k}]| \leq  \begin{cases}
   0  &\text{if $k=1$ or $2$,} \\
    c_k{n\choose q_n}^{-1/2}  &  \text{if $k\geq 3$ is odd,}\\
    c'_k{n\choose q_n}^{-1 }  &  \text{if $k\geq 4$ is even.}
 \end{cases}$$
 \end{lem}
We note that if $1\leq q_n\leq n-1$, then   ${n\choose q_n}^{-1/2}\to 0$. If $k$ is odd, then $P_2(I_n^k)=\emptyset$ by definition, and thus  the second inequality in Lemma \ref{lemma1} together with \eqref{tr}   imply that
\begin{equation}
  \label{odd}
  \begin{aligned}
  |m_{n,k}|&\leq  \frac 1{{n\choose q_n}^{k/2}} \sum_{{R_1},...,  {R_k}\in I_n} |\e[J_{R_1}\cdots J_{R_k}]| \frac 1{L_n}  |\Tr \Psi_{R_1}\cdots \Psi_{R_k}|\\
 &\leq c_k{n\choose q_n}^{-1/2} \to 0,
  \end{aligned}
\end{equation}
i.e.,  all the odd moments tend to 0. All of the rest is to estimate the moment $m_{n,k}$ when $k$ is even.


\subsubsection{Proof of case 1}
Now we are ready to prove case 1 in Theorem \ref{main1}.  For $k$ even,
similarly, Lemma \ref{lemma1} also implies that the second summation in \eqref{deco}  satisfies
$$| \sum_{({R_1},...,  {R_k})\in I^k_n\setminus P_2(I_n^k)}|\to 0.$$
To deal with the first summation, we further define the set
 $$A_{n,k}=\{({R_1},...,  {R_k})\in P_2( I_n^k) | R_i\cap R_j=\emptyset\,\,\mbox{if}\,\, R_i\neq   R_j \},\,\,\,B_{n,k}=P_2( I_n^k)\setminus A_{n,k}. $$ Then the first summation  can be further rewritten as
\begin{equation}\label{dec} \sum_{({R_1},...,  {R_k})\in  P_2(I_n^k)}=\sum_{A_{n,k}}+\sum_{B_{n,k}}.\end{equation}
We make such decomposition because of the following identity
\begin{equation}\label{deddc}  i^{q_nk/2}\Psi_{R_1}\cdots \Psi_{R_k}=I \,\,\,\mbox{for}\,\,\,({R_1},...,  {R_k})\in  A_{n,k}.\end{equation}
 \begin{lem}\label{lemma2}If  $q_n\ll\sqrt n$, for $k$ even,  we have the bounds
 $$     {n\choose q_n}^{-k/2}|B_{n,k}| \leq c_k (k-1)!!\frac{q_n^2}{n}$$
and   $$(k-1)!!\left(1-  c_k\frac{q_n^2}{n} \right)\leq  {n\choose q_n}^{-k/2}|A_{n,k}|\leq (k-1)!!,$$
where $c_k= \frac {k^2  }{4 }$.  
\end{lem}

 Lemma \ref{lemma2}   implies that    ${n\choose q_n}^{-k/2}|B_{n,k}| \to 0 $ and
${n\choose q_n}^{-k/2}|A_{n,k}|\to (k-1)!!$   for $q_n\ll\sqrt n$. Hence,  by \eqref{exp}\eqref{deddc}, we have the limit
$$\sum_{A_{n,k}}= \frac 1{L_n} \frac 1{{n\choose q_n}^{k/2}} \sum_{A_{n,k}}   \Tr I={n\choose q_n}^{-k/2}|A_{n,k}|\to (k-1)!!.$$
By \eqref{tr} again, we have the estimate
$$|\sum_{B_{n,k}}|\leq    \frac 1{{n\choose q_n}^{k/2}}  \sum_{B_{n,k}} 1 ={n\choose q_n}^{-k/2}|B_{n,k}|\to 0.$$
 To summarize, for  $k$ even, we have
\begin{equation}\label{even}m_{n,k}= \sum_{({R_1},...,  {R_k})\in  I_n^k\setminus P_2(I_n^k)}+\sum_{A_{n,k}}+\sum_{B_{n,k}}\to (k-1)!!.\end{equation}
Equations \eqref{odd}\eqref{even} show that $m_{n,k}$ is asymptotic to the $k$-th moment of the standard Gaussian distribution which satisfies Carleman's condition, and thus the expectation of $\rho_n$ will tend to the standard Gaussian distribution by the moment method.
\subsubsection{Proof of case 2}\label{222}
To prove case 2 for $k$ even, we need to treat the summation over $P_2(I_n^k)$   in a different way from case 1.
Let's define the set of $2$ to $1$ maps as \begin{equation}\label{sk}S_k=\left\{\pi: \,\,\{1,2,\cdots, k\}\to \{1,2,\cdots, \frac k2\}| |\pi^{-1}(j)|=2,  1\leq j\leq \frac k2  \right\}.\end{equation}
The cardinality of this set is $|S_k|=(k-1)!!\cdot(k/2)!=k!/2^{k/2}$.

Then  we can rewrite,
\begin{align*}\sum _ {P_2(I_n^k)}&= \frac 1{L_n} \frac {i^{q_nk/2}}{{n\choose q_n}^{k/2}} \sum_{\pi\in S_k} \sum_{R_1,\cdots, R_{\frac k2}\in I_n, R_i \neq R_j \,\mbox{if}\,\, i\neq j}\frac{\e[J_{R_{\pi(1)}}\cdots J_{R_{\pi(k)}}]  \Tr \Psi_{R_{\pi(1)}}\cdots \Psi_{R_{\pi(k)}}}{(k/2)!}\\
&= \frac 1{L_n} \frac {i^{q_nk/2}}{{n\choose q_n}^{k/2}} \sum_{\pi\in S_k} \sum_{R_1,\cdots, R_{\frac k2}\in I_n, R_i \neq R_j \,\mbox{if}\,\, i\neq j}   \frac{\Tr \Psi_{R_{\pi(1)}}\cdots \Psi_{R_{\pi(k)}}}{(k/2)!},
\end{align*}
where we use the fact that $\e[J_{R_{\pi(1)}}\cdots J_{R_{\pi(k)}}]=1$ by \eqref{exp}.

We now introduce the notion of crossing number $\kappa (\pi)$ for a pair-partition $\pi$, which is defined to be the number of subsets
$\{r,s\}\subset  \{1,2,\cdots, \frac k2\}$ such that there exists $1\leq a<b<c<d\leq k, \pi(a)=\pi(c)=r,  \pi(b)=\pi(d)=s$. Let $\{r_1,s_1\},\{r_2,s_2\},\cdots, \{r_{\kappa(\pi)},s_{\kappa(\pi)}\}$ be the crossings of $\pi$. By  \eqref{anticom}   , we easily have
\begin{lem} \label{lemma3}For the fixed $\pi$, we have $$\frac {i^{q_nk/2}} {L_n} \Tr  \Psi_{R_{\pi(1)}}\cdots \Psi_{R_{\pi(k)}}=(-1)^{\sum_{k=1}^{\kappa(\pi)}|R_{r_k}\cap R_{s_k}|}.$$
\end{lem}
If $\pi$ has no crossing,  it reads $$\frac {i^{q_nk/2}}  {L_n} \Tr  \Psi_{R_{\pi(1)}}\cdots \Psi_{R_{\pi(k)}}=1.$$

The following lemma deals with the cardinality of the intersection of the coordinates $|R_{r_k}\cap R_{s_k}|$ where we refer to \cite{So} for the proof.
 \begin{lem}\label{lemma4}When $ q_n^2/n\to a$, if we choose $\{R_1,\cdots, R_{\frac k2}\}$ uniformly from $I_n^{\frac k 2}$ with $R_i\neq R_j$ if $i\neq j$, then the intersection numbers $|R_{r_k}\cap R_{s_k}|,  k=1,\cdots , \kappa(\pi)$ are  approximately independently Poisson(a) distributed. Here, $\{ r_k, s_k\}_{k=1}^{\kappa(\pi)}$ are crossings of $\pi$.

\end{lem}
With Lemma \ref{lemma3} and Lemma \ref{lemma4}, for any fixed  map $\pi$, we have
\begin{align*}&\lim_{n\to \infty}  \frac 1{L_n} \frac {i^{q_nk/2}}{{n\choose q_n}^{k/2}} \sum_{R_1,\cdots, R_{\frac k2}\in I_n, R_i \neq R_j \,\mbox{if}\,\, i\neq j}   \Tr \Psi_{R_{\pi(1)}}\cdots \Psi_{R_{\pi(k)}}\\
&=\lim_{n\to \infty}   \frac 1{{n\choose q_n}^{k/2}} \sum_{R_1,\cdots, R_{\frac k2}\in I_n, R_i \neq R_j \,\mbox{if}\,\, i\neq j}  (-1)^{\sum_{k=1}^{\kappa(\pi)}|R_{r_k}\cap R_{s_k}|} \\
&=\sum_{m_i\geq 0, 1\leq i\leq \kappa(\pi)}(-1)^{m_1+\cdots+m_{\kappa(\pi)}}\frac{a^{m_1+\cdots+m_{\kappa(\pi)}}}{m_1!\cdots m_{\kappa(\pi)}!}e^{-a\kappa(\pi)}\\
&=e^{-2a\kappa(\pi)}.
\end{align*}
Therefore, we have
\begin{equation}\label{mka}\lim_{n\to \infty}\sum _ {P_2(I_n^k)}
 = \frac{1}{(k/2)!}\sum_{\pi\in S_k} e^{-2a\kappa(\pi)}:=m_k^a.
\end{equation}
Defining $m_k^a=0 $ for $k$ odd, then we have proved
$$\lim_{n\to\infty} m_{n,k}=
   m^a_k .$$

The theory of $q$-Hermite polynomials \cite{So, ISV} implies that the  moments $m_k^a$ correspond to the density function $p_a(x)$ given in Theorem \ref{main1}.  $p_a(x)$ has compact support,  and thus  $p_a(x)$ satisfies Carleman's condition. Therefore,  the moment method implies that the expectation of the normalized empirical measure $\mathbb \rho_n$ indeed converges to $p_a(x)dx$ as $n\to \infty$ for $a\in(0,\infty)$. In fact, $m_k^a $ is also well defined for $a=0$ where it's easy to get $m_k^0=(k-1)!!$ if $k$ is even and $m_k^0=0$ if $k$ is odd, which is the $k$-th moment of the standard Gaussian distribution, and the method for case 2 can be also  used to prove case 1.

\subsubsection{Proof of case 3}
 Heuristically, it seems that if we take $a_n:=q_n^2/n\to \infty$ in \eqref{mka} where the summation $\sum _ {P_2(I_n^k)}$ is approximated to $ \frac{1}{(k/2)!}\sum_{\pi\in S_k} e^{-2a_n\kappa(\pi)}$, then only the non-crossing $\pi$  with $\kappa(\pi)=0$ contributes to this summation as $n\to\infty$. By counting the number of such non-crossing $\pi$, we have $$m_{n,k}\to \frac{k!}{( k/2)!(  k/2+1)!},$$ i.e.,   the Catalan numbers, which is the $k$-th moment of the semicircle law.
The strategy to prove the above argument rigorously is to prove
\begin{lem}\label{lemma5}For the fixed map $\pi$, let's assume $\kappa(\pi)\geq 1$. If $q_n^2/n\to\infty$ and $q_n\leq n/2$, if we choose $R_j$ uniformly, then asymptotically we have
 \begin{align*}\lim_{n\to\infty}\mathbb E[(-1)^{\sum_{k=1}^{\kappa(\pi)}|R_{r_k}\cap R_{s_k}|}]=0.\end{align*}
\end{lem}Here, the condition $R_i\neq R_j$ if $i\neq j$ is omitted since the probability of $R_i= R_j$ tends to $0.$
By Lemma \ref{lemma5} and following the arguments in \S \ref{222} above,  we know that only non-crossing $\pi$ with $\kappa(\pi)=0$ contributes to the summation $\sum _ {P_2(I_n^k)}$ , this will yield the Catalan number and  hence the semicircle law.
\subsection{Proof of lemmas}\label{lemmas}
In this subsection, let's prove Lemmas \ref{lemma1},  \ref{lemma2} and \ref{lemma5}.
\subsubsection{Proof of Lemma 1}

\begin{proof}
By assumption, $J_{R}$ are independent with mean $0$ and variance $1$, i.e., $\e J_{R}=0$ and  $\e J_{R_i}J_{R_j}=0$ if $i\neq j$, this implies the case when $k=1$ or $2$ in Lemma \ref{lemma1}.

For $k\geq 3$, by assumption that  $J_{R_i}$ have uniformly bounded moments, i.e.,  for any fixed $k$, there exists $c_k$ such that $|\e J_{R}^k|\leq c_k$ for all $R$ in $I_n$, we will easily have  $|\e J_{R_1}\cdots J_{R_k} |\leq c_k$. Furthermore, if some $R_i$ appears only once in $(R_1,\cdots, R_k)$, then $\e J_{R_1}\cdots J_{R_k} =0$. Hence,
\begin{align*}\frac 1{{n\choose q_n}^{k/2}}  \sum_{({R_1},...,  {R_k})\in I^k_n\setminus P_2(I_n^k)}| \e[J_{R_1}\cdots J_{R_k}]| \leq \frac {c_k}{{n\choose q_n}^{k/2}}  \sum_{({R_1},...,  {R_k})\in I^k_n\setminus P_2(I_n^k), \# R_i\geq 2 } 1. \end{align*}
Let's denote $k_n:=|\{({R_1},...,  {R_k})\in I^k_n\setminus P_2(I_n^k), \# R_i\geq 2 \}|$, where $ \#A:=|\{j|1\leq j\leq k,R_j=A\}|$ for $A\in I_n$. Then for $({R_1},...,  {R_k})\in I^k_n$ with $\# R_i\geq 2 $, we have\begin{align*} k=\sum_{A\in\{R_j|1\leq j\leq k\}}\#A\geq\sum_{A\in\{R_j|1\leq j\leq k\}}2=2|\{R_j|1\leq j\leq k\}|. \end{align*} The equality holds if and only if $({R_1},...,  {R_k})\in  P_2(I_n^k)$. Thus,  if $({R_1},...,  {R_k})\in I^k_n\setminus P_2(I_n^k)$ with  $\# R_i\geq 2 $,  then $|\{R_j|1\leq j\leq k\}|<k/2.$

 If $k\geq 3$ is odd, then $|\{R_j|1\leq j\leq k\}|\leq (k-1)/2 $.  For $n$ large enough, we have \begin{align*} k_n&\leq \sum_{B\subseteq I_n,|B|=(k-1)/2}|\{({R_1},...,  {R_k})\in I^k_n|  R_i\in B,\ \forall\ 1\leq i\leq k \}|\\&=\sum_{B\subseteq I_n,|B|=(k-1)/2}\left(\frac{k-1}{2}\right)^k=\left(\frac{k-1}{2}\right)^k{|I_n|\choose {\frac {k-1}2}}\leq c_k |I_n|^{\frac {k-1}2}, \end{align*}
where $|I_n|={n\choose q_n}$. This further implies
$$ \frac {c_k}{{n\choose q_n}^{k/2}}  \sum_{({R_1},...,  {R_k})\in I^k_n\setminus P_2(I_n^k), \# R_i\geq 2 } 1 \leq  \frac {c_k}{{n\choose q_n}^{k/2}} |I_n|^{\frac {k-1}2}=c_k{n\choose q_n}^{-\frac 12},$$
 which completes the case when $k\geq 3$ is odd.

 Similarly, if $k\geq 4$ is even, then $|\{R_j|1\leq j\leq k\}|\leq k/2-1 $ and
 \begin{align*} k_n\leq\sum_{B\subseteq I_n,|B|=k/2-1}|\{({R_1},...,  {R_k})\in I^k_n|  R_i\in B,\ \forall\ 1\leq i\leq k \}|\leq c_k |I_n|^{\frac {k}2-1}. \end{align*}
  This implies $$ \frac {c_k}{{n\choose q_n}^{k/2}}  \sum_{({R_1},...,  {R_k})\in I^k_n\setminus P_2(I_n^k), \# R_i\geq 2 } 1 \leq  \frac {c_k}{{n\choose q_n}^{k/2}} |I_n|^{\frac {k}2-1}=c_k{n\choose q_n}^{-1},$$
 which completes the case when $k\geq 4$ is even.
\end{proof}
\subsubsection{Proof of Lemma 2}
\begin{proof}
We first have the upper bound $|A_{n,k}|\leq |P_2(I_n^k)|\leq|I_n|^{\frac k 2}(k-1)!!={n\choose q_n}^{\frac k2}(k-1)!!$. For the lower bound, we have
\begin{align*} |A_{n,k}|&={n\choose q_n}{n- q_n\choose q_n}\cdots {n- (\frac k2 -1)q_n\choose q_n}(k-1)!!\\
&\geq \left(\frac{{n- (\frac k2 -1)q_n\choose q_n}}{{n\choose q_n}}\right)^{\frac k2}{n\choose q_n}^{\frac k2}(k-1)!!\\
&\geq \left(\frac{n-\frac k2 q_n}n\right)^{\frac k 2 q_n}{n\choose q_n}^{\frac k2}(k-1)!!\\
&\geq \left(1-\frac {k^2 q_n^2}{4n}\right) {n\choose q_n}^{\frac k2}(k-1)!!,
\end{align*}
where in the last inequality, we use the inequality $(1-x)^t\geq 1-tx$ for $0<x<1$ and $t\geq 1$.
By the lower bound of $|A_{n,k}|$, we have the upper bound
$$|B_{n,k}|= |P_2(I_n^k)|-|A_{n,k}|\leq \frac {k^2 q_n^2}{4n} (k-1)!! {n\choose q_n}^{\frac k2},$$
this finishes Lemma \ref{lemma2}.
\end{proof}


\subsubsection{Proof of Lemma \ref{lemma5}}
\begin{proof}
Let's assume $R_1,\cdots, R_l$ are mutually independent. For $1\leq i<j\leq l,$
let $S_{i,j}=\cup_{k\neq i,j}R_k,\ T_{i,j}=R_i\cap S_{i,j}$ and  $V_{i,j}=R_i\setminus S_{i,j}.$ Let $ \mathcal{F}_{i,j}$ be the $ \sigma$-algebra generated by $R_k$ for every $k\neq i,j$ and $T_{i,j},\, T_{j,i}. $ The whole proof is based on the estimate of the following conditional expectation\begin{equation}\label{lll}\l_{i,j}=\mathbb E[(-1)^{|R_i\cap R_{j}|} \big|  \mathcal{F}_{i,j}]=(-1)^{|R_i\cap R_{j}\cap S_{i,j}|}\mathbb E[(-1)^{|V_{i,j}\cap V_{j,i}|} \big|  \mathcal{F}_{i,j}].
\end{equation}Conditioning on $\mathcal{F}_{i,j} $, $V_{j,i}$ is uniformly distributed among all the subsets of $ I_n\setminus S_{i,j}$ with exactly $|V_{j,i}|$ elements (notice that $|V_{j,i}|=|R_i|-|T_{j,i}|$ is measurable in $ \mathcal{F}_{i,j}$),  $V_{i,j}$ and $V_{j,i}$ are conditionally independent. Then we have\begin{align*} \mathbb P(|V_{i,j}\cap V_{j,i}|=k \big| \mathcal{F}_{i,j})={|V_{i,j}|\choose {k }}{|I_n\setminus S_{i,j}|-|V_{i,j}|\choose {|V_{j,i}|-k }}/{{|I_n\setminus S_{i,j}|\choose |V_{j,i}|}}
\end{align*} and\begin{align*} \mathbb E[(-1)^{|V_{i,j}\cap V_{j,i}|} \big|  \mathcal{F}_{i,j}]=F(|V_{i,j}|,|V_{j,i}|,n-|S_{i,j}|)
\end{align*} where \begin{align*}F(p,q,m)={\sum_k(-1)^k{p\choose {k }}{m-p\choose {q-k }}}/{{m\choose q}}
\end{align*} for $p,q,m\in\mathbb{Z}$ and $p,q\in [0,m]$. Regarding $F(p,q,m)$, we further have   \begin{equation}\label{rela}F(p,q,m)=F(q,p,m)=(-1)^qF(m-p,q,m)=(-1)^pF(p,m-q,m),
\end{equation}which indicates that if we want to estimate $F(p,q,m)$, it's enough to consider the case $p,q\in [0,m/2]$. Since $ F(p,q,m){m\choose q}$ is the coefficient of $x^q$ in the polynomial $(1-x)^p(1+x)^{m-p}=(1-x^2)^p(1+x)^{m-2p},$ we have
$$F(p,q,m)={\sum_k(-1)^k{p\choose {k }}{m-2p\choose {q-2k }}}/{{m\choose q}}.$$
If $m\geq2(p+q)$, then we have $q-2k\leq q-k$ and $ (q-2k)+(q-k)\leq m-2p$ for $k\geq0$, and thus  ${m-2p\choose {q-2k }}\leq {m-2p\choose {q-k }} $.

 Hence, for $m\geq2(p+q)$, we have
\begin{align*}|F(p,q,m)|&\leq{\sum_k{p\choose {k }}{m-2p\choose {q-2k }}}/{{m\choose q}}
\leq{\sum_k{p\choose {k }}{m-2p\choose {q-k }}}/{{m\choose q}}
\\&={m-p\choose {q }}/{{m\choose q}}=\prod_{k=0}^{q-1}\frac{m-p-k}{m-k}\\
&\leq\prod_{k=0}^{q-1}\frac{m-p}{m}=\left(\frac{m-p}{m}\right)^q\leq e^{-\frac{pq}{m}}.
\end{align*}If $m\leq2(p+q),$ let $r=\lceil\frac{1}{2}(\frac{m}{2}-p+q)\rceil$, i.e., $r-1<\frac{1}{2}(\frac{m}{2}-p+q)\leq r,$ then we have $q/2\leq r\leq q.$  For any integer $k\in [0,p]$ such that $0\leq q-2k\leq m-2p,$ there are two cases we need to consider:

\textcircled{1} If $r-k\leq q-2k$, then $k\leq q-r$ and $(r-k)+( q-2k)=r+q-3k\geq r+q-3(q-r)=4r-2q\geq4\cdot\frac{1}{2}\left(\frac{m}{2}-p+q\right)-2q=m-2p$, hence, we have ${m-2p\choose {q-2k }}\leq {m-2p\choose {r-k }}$.

\textcircled{2} If $r-k> q-2k$, then $k\geq q-r+1$ and $(r-k)+( q-2k)=r+q-3k\leq r+q-3(q-r+1)=4(r-1)-2q+1
<4\cdot\frac{1}{2}\left(\frac{m}{2}-p+q\right)-2q+1=m-2p+1$ and $(r-k)+( q-2k)\leq m-2p$, and
thus we also have ${m-2p\choose {q-2k }}\leq {m-2p\choose {r-k }}$.

Therefore, by combining cases \textcircled{1}  \textcircled{2}, for $m\leq2(p+q),$ we always have
\begin{align*}|F(p,q,m)|&\leq{\sum_k{p\choose {k }}{m-2p\choose {q-2k }}}/{{m\choose q}}
\leq{\sum_k{p\choose {k }}{m-2p\choose {r-k }}}/{{m\choose q}}
\\&={m-p\choose {r }}/{{m\choose q}}\leq{m-p\choose {r }}/{{m\choose r}}\leq e^{-\frac{pr}{m}}\leq e^{-\frac{pq}{2m}}.
\end{align*}For the general case $p,q\in [0,m],$ by \eqref{rela} and the estimates above,  we finally have \begin{equation}\label{esss}
	|F(p,q,m)|=|F(\min(p,m-p),\min(q,m-q),m)|\leq e^{-\frac{\min(p,m-p)\min(q,m-q)}{2m}}.
\end{equation}We denote $$a_{i,j}:=\min(n-|S_{i,j}|-|V_{i,j}|,|V_{i,j}|).$$ By defintion  \eqref{lll} and the estimate \eqref{esss},
we have \begin{equation}\label{lambda}|\l_{i,j}|=|F(|V_{i,j}|,|V_{j,i}|,n-|S_{i,j}|)|\leq e^{-\frac{a_{i,j}a_{j,i}}{2(n-|S_{i,j}|)}}.
\end{equation}All the rest is to derive the bounds of $a_{i,j} $ and $ a_{j,i}$ in probability in order to control $|\lambda_{i,j}|$, where we  need to estimate the expectation and variance of $n-|S_{i,j}|-|V_{i,j}|$ and $|V_{i,j}|$ first. We notice that $n-|S_{i,j}|-|V_{i,j}|=n-|S_{i,j}\cup V_{i,j}|=n-|S_{i,j}\cup R_{i}|=n-|\cup_{k\neq j}R_k|,$ then we have\begin{align*}&\,\,\,\mathbb E[n-|S_{i,j}|-|V_{i,j}|]=\sum_{m=1}^n\mathbb P(m\not\in R_k,\forall\ k\neq j)\\&=n\left({n-1\choose {q_n }}/{{n\choose q_n}}\right)^{l-1}=n\left(1-\frac{q_n}{n}\right)^{l-1}
\end{align*} and \begin{align*}&\,\,\,\mathbb E{n-|S_{i,j}|-|V_{i,j}|\choose 2}=\sum_{1\leq p<q\leq n}\mathbb P(p,q\not\in R_k,\forall\ k\neq j)\\&={n\choose 2}\left({n-2\choose {q_n }}/{{n\choose q_n}}\right)^{l-1}={n\choose 2}\left(1-\frac{q_n}{n}\right)^{l-1}\left(1-\frac{q_n}{n-1}\right)^{l-1}.
\end{align*}This implies that the variance \begin{align*}&\,\,\,\,\,\,\text{var}[n-|S_{i,j}|-|V_{i,j}|]\\&=\mathbb E\left|n-|S_{i,j}|-|V_{i,j}|\right|^2-\left|n\left(1-{q_n}/{n}\right)^{l-1}\right|^2\\&=2\mathbb E{n-|S_{i,j}|-|V_{i,j}|\choose 2}+\mathbb E[n-|S_{i,j}|-|V_{i,j}|]-n^2\left(1-\frac{q_n}{n}\right)^{2(l-1)}\\
&={n(n-1)}\left(1-\frac{q_n}{n}\right)^{l-1}\left(1-\frac{q_n}{n-1}\right)^{l-1} +n\left(1-\frac{q_n}{n}\right)^{l-1}-n^2\left(1-\frac{q_n}{n}\right)^{2(l-1)}\\&={n}\left(1-\frac{q_n}{n}\right)^{l-1}
\left((n-1)\left(1-\frac{q_n}{n-1}\right)^{l-1} +1-n\left(1-\frac{q_n}{n}\right)^{l-1}\right)\\&\leq{n}\left(1-\frac{q_n}{n}\right)^{l-1}
\left(1-\left(1-\frac{q_n}{n}\right)^{l-1}\right)\leq{n}\left(1-\frac{q_n}{n}\right)^{l-1}.
\end{align*} Therefore, we further have (using $q_n\leq n/2$) \begin{equation}\label{p1}\begin{aligned}&\,\,\,\mathbb P\left(n-|S_{i,j}|-|V_{i,j}|\leq \frac{n}{2}\left(1-\frac{q_n}{n}\right)^{l-1}\right)\\&\leq\text{var}[n-|S_{i,j}|-|V_{i,j}|]
\left|\frac{n}{2}\left(1-\frac{q_n}{n}\right)^{l-1}\right|^{-2}\\&\leq{n}\left(1-\frac{q_n}{n}\right)^{l-1}
\left|\frac{n}{2}\left(1-\frac{q_n}{n}\right)^{l-1}\right|^{-2}\\&
=\frac{4}{n}\left(1-\frac{q_n}{n}\right)^{-(l-1)}\leq\frac{2^{l+1}}{n}.
\end{aligned}\end{equation}Similarly, we have  \begin{align*}\mathbb E[|V_{i,j}|]&=\sum_{m=1}^n\mathbb P(m\in V_{i,j})=\sum_{m=1}^n\mathbb P(m\in R_i,m\not\in R_k,\forall\ k\neq i,j)\\&=n{n-1\choose {q_n-1 }}/{{n\choose q_n}}\left({n-1\choose {q_n }}/{{n\choose q_n}}\right)^{l-2}=q_n\left(1-\frac{q_n}{n}\right)^{l-2}
\end{align*} and \begin{align*}\mathbb E{|V_{i,j}|\choose 2}&=\sum_{1\leq p<q\leq n}\mathbb P(p,q\in V_{i,j})\\&={n\choose 2}{n-2\choose {q_n-2 }}/{{n\choose q_n}}\left({n-2\choose {q_n }}/{{n\choose q_n}}\right)^{l-2}\\&={q_n\choose 2}\left(1-\frac{q_n}{n}\right)^{l-2}\left(1-\frac{q_n}{n-1}\right)^{l-2}.
\end{align*}Hence, we have \begin{align*}\text{var}[|V_{i,j}|]&=\mathbb E|V_{i,j}|^2-\left|\mathbb E[|V_{i,j}|]\right|^2 =2\mathbb E{|V_{i,j}|\choose 2}+\mathbb E[|V_{i,j}|]-\left|\mathbb E[|V_{i,j}|]\right|^2\\&={q_n}\left(1-\frac{q_n}{n}\right)^{l-2}
\left((q_n-1)\left(1-\frac{q_n}{n-1}\right)^{l-2} +1-q_n\left(1-\frac{q_n}{n}\right)^{l-2}\right)\\&\leq{q_n}\left(1-\frac{q_n}{n}\right)^{l-2}
\left(1-\left(1-\frac{q_n}{n}\right)^{l-2}\right)\leq{q_n}\left(1-\frac{q_n}{n}\right)^{l-2}
(l-2)\frac{q_n}{n}.
\end{align*}Therefore, we can derive \begin{equation}\label{p2}\begin{aligned}  &\mathbb P\left(|V_{i,j}|\leq \frac{q_n}{2}\left(1-\frac{q_n}{n}\right)^{l-2}\right)\leq\text{var}[|V_{i,j}|]
\left|\frac{q_n}{2}\left(1-\frac{q_n}{n}\right)^{l-2}\right|^{-2}\leq\frac{2^l(l-2)}{n}.
\end{aligned}\end{equation}Since $a_{i,j}=\min(|V_{i,j}|,n-|S_{i,j}|-|V_{i,j}|)$,  by the estimates $$\frac{q_n}{2^{l-1}}\leq\frac{q_n}{2}\left(1-\frac{q_n}{n}\right)^{l-2}\leq \frac{n-q_n}{2}\left(1-\frac{q_n}{n}\right)^{l-2}=\frac{n}{2}\left(1-\frac{q_n}{n}\right)^{l-1}$$ and the estimates \eqref{p1}\eqref{p2}, we   have\begin{align*}&\,\,\,\,\,\,\,\,\mathbb P\left(a_{i,j}\leq q_n/2^{l-1}\right)\\&\leq\mathbb P\left(|V_{i,j}|\leq \frac{q_n}{2}\left(1-{q_n}/{n}\right)^{l-2}\right)+\mathbb P\left(n-|S_{i,j}|-|V_{i,j}|\leq \frac{n}{2}\left(1-{q_n}/{n}\right)^{l-1}\right)\\&\leq\frac{2^{l+1}}{n}+\frac{2^l(l-2)}{n}=\frac{2^ll}{n}.
\end{align*}Recall  \eqref{lambda}, we finally have  \begin{align*}\mathbb E|\l_{i,j}|&\leq\mathbb E e^{-\frac{a_{i,j}a_{j,i}}{2(n-|S_{i,j}|)}}\\&\leq
e^{-\frac{(q_n/2^{l-1})^2}{2n}}+
\mathbb P(a_{i,j}\leq q_n/2^{l-1})
+\mathbb P(a_{j,i}\leq q_n/2^{l-1})\\ &\leq e^{-2^{1-2l}\frac{q_n^2}{n}}+\frac{2^ll}{n}+\frac{2^ll}{n}.
\end{align*}
Now we are ready to finish the proof of Lemma \ref{lemma5} where we  replace $l:=\kappa(\pi)$.
We notice that if $k\neq i,j$, then $R_i\cap R_k=T_{i,j}\cap R_k$ and $|R_i\cap R_k| $ are measurable in $ \mathcal{F}_{i,j}$, so is  $|R_j\cap R_k|.$ If  $\{k,m\}\subseteq I_n\setminus\{ i,j\}$, then $R_m\cap R_k$ and $|R_m\cap R_k| $ are measurable in $ \mathcal{F}_{i,j}.$ Therefore, if $q_n^2/n\to \infty$, $q_n \leq n/2 $ and $\kappa(\pi)>0$, we have\begin{align*}&\,\,\,\,\,\,\,\left|\mathbb E\left[(-1)^{\sum_{k=1}^{\kappa(\pi)}|R_{r_k}\cap R_{s_k}|}\right]\right|\\&=\left|\mathbb E\left[\mathbb E\left[(-1)^{\sum_{k=1}^{\kappa(\pi)}|R_{r_k}\cap R_{s_k}|}|\mathcal{F}_{r_1,s_1}\right]\right]\right|\\&=\left|\mathbb E\left[(-1)^{\sum_{k=2}^{\kappa(\pi)}|R_{r_k}\cap R_{s_k}|}\mathbb E\left[(-1)^{|R_{r_1}\cap R_{s_1}|}|\mathcal{F}_{r_1,s_1}\right]\right]\right|\\&=\left|\mathbb E\left[(-1)^{\sum_{k=2}^{\kappa(\pi)}|R_{r_k}\cap R_{s_k}|}\l_{r_1,s_1}\right]\right|\\&\leq\mathbb E|\l_{r_1,s_1}|\\&\leq e^{-2^{1-2l}\frac{q_n^2}{n}}+\frac{2^{l+1}l}{n}.\end{align*}
which implies \begin{align*}\lim_{n\to\infty}\mathbb E[(-1)^{\sum_{k=1}^{\kappa(\pi)}|R_{r_k}\cap R_{s_k}|}]=0,\end{align*}
which finishes the proof of Lemma \ref{lemma5}.\end{proof}
\section{Variance and almost sure convergence}\label{almost}
In this section, we will derive an upper bound about the variance of $L_n^{-1}\Tr H^k$. As a direct consequence,  we will prove that $\rho_n\to \rho_\infty$ almost surely in the sense of distribution. Combining this with the results we derived in \S\ref{expected},  we finish the proof of Theorem \ref{main1}.
\begin{lem}\label{lemma6}Let $q_n$ be even and $2\leq q_n<n$,  then $$   \text{var}[L_n^{-1} \Tr H^k]\leq c_k{n\choose q_n}^{-1}$$ for some constant  $c_k$ for any $ k\geq1$.
\end{lem}\begin{proof}Since \begin{equation}\label{mnk1} \frac 1{L_n}    \Tr H^k= \frac 1{L_n} \frac {i^{q_nk/2}}{{n\choose q_n}^{k/2}} \sum_{{R_1},...,  {R_k}\in I_n} J_{R_1}\cdots J_{R_k}  \Tr \Psi_{R_1}\cdots \Psi_{R_k}, \end{equation} we have \begin{align*} \text{var}[L_n^{-1} \Tr H^k]=& \frac 1{L_n^2} \frac {(-1)^{q_nk/2}}{{n\choose q_n}^{k}} \sum_{{R_1},...,  {R_{2k}}\in I_n} \text{cov}(J_{R_1}\cdots J_{R_k},J_{R_{k+1}}\cdots J_{R_{2k}})\cdot\\ & \Tr \Psi_{R_1}\cdots \Psi_{R_k}\Tr \Psi_{R_{k+1}}\cdots \Psi_{R_{2k}}. \end{align*}
For every ${R_1},...,  {R_{2k}}\in I_n $ and $A\in I_n$, let's denote $ \#A=|\{j|1\leq j\leq 2k,R_j=A\}|.$

Since $J_{R_i}$ have uniformly bounded moments, we have  \begin{align*}|\text{cov}(J_{R_1}\cdots J_{R_k},J_{R_{k+1}}\cdots J_{R_{2k}})|&\leq \text{var}[J_{R_1}\cdots J_{R_k}]^{\frac{1}{2}}\text{var}[J_{R_{k+1}}\cdots J_{R_{2k}}]^{\frac{1}{2}}\\ &\leq(\e|J_{R_1}\cdots J_{R_k}|^2)^{\frac{1}{2}}(\e|J_{R_{k+1}}\cdots J_{R_{2k}}|^2)^{\frac{1}{2}}\\&\leq c_{2k}.\end{align*} Furthermore, if some $R_i$  appears only once in $(R_1,\cdots, R_{2k})$, then $\e J_{R_1}\cdots J_{R_{2k}} =0$, $\e J_{R_1}\cdots J_{R_{k}} =0 $ if $1\leq i\leq k$ and $\e J_{R_{k+1}}\cdots J_{R_{2k}} =0 $ if $k+1\leq i\leq 2k,$ thus,
 $\text{cov}(J_{R_1}\cdots J_{R_k},J_{R_{k+1}}\cdots J_{R_{2k}})=\e J_{R_1}\cdots J_{R_{2k}}-\e J_{R_1}\cdots J_{R_{k}}\e J_{R_{k+1}}\cdots J_{R_{2k}} =0 $.
Hence, we can write\begin{align*} \text{var}[L_n^{-1} \Tr H^k]&= \frac 1{L_n^2} \frac {(-1)^{q_nk/2}}{{n\choose q_n}^{k}} \left(\sum_{({R_1},...,  {R_{2k}})\in P_2(I_n^{2k})}+ \sum_{({R_1},...,  {R_{2k}})\in I_n^{2k}\setminus P_2(I_n^{2k}), \# R_i\geq 2}\right)\\ &\text{cov}(J_{R_1}\cdots J_{R_k},J_{R_{k+1}}\cdots J_{R_{2k}})\cdot  \Tr \Psi_{R_1}\cdots \Psi_{R_k}\Tr \Psi_{R_{k+1}}\cdots \Psi_{R_{2k}}\\&:=V_1+V_2. \end{align*}Let's denote $k_n':=|\{({R_1},...,  {R_{2k}})\in I^{2k}_n\setminus P_2(I_n^{2k}), \# R_i\geq 2 \}|$.
If $({R_1},...,  {R_{2k}})\in I^{2k}_n\setminus P_2(I_n^{2k})$ with $\# R_i\geq 2 $, then $|\{R_j|1\leq j\leq 2k\}|\leq k-1$. For $n$ large enough, we have \begin{align*} k_n'&\leq\sum_{B\subseteq I_n,|B|=k-1}|\{({R_1},...,  {R_{2k}})\in I^{2k}_n|  R_i\in B,\ \forall\ 1\leq i\leq 2k \}|\\&=\sum_{B\subseteq I_n,|B|=k-1}({k-1})^{2k}=({k-1})^{2k}{|I_n|\choose { {k-1}}}\leq c_k |I_n|^{ {k-1}}. \end{align*}
   By \eqref{tr},  we have\begin{align*}|V_2|\leq \frac 1{{n\choose q_n}^{k}} \sum_{({R_1},...,  {R_{2k}})\in I_n^{2k}\setminus P_2(I_n^{2k}), \# R_i\geq 2} c_{2k}\leq  \frac {c_k}{{n\choose q_n}^{k}} |I_n|^{k-1}=c_k{n\choose q_n}^{-1}. \end{align*}
   Now we turn to estimate $V_1$.  For $({R_1},...,  {R_{2k}})\in P_2(I_n^{2k}) $, we denote $A_1:=\{R_j|1\leq j\leq k\},\ A_2:=\{R_j|k+1\leq j\leq 2k\}$ and  $A_0:=A_1\cap A_2$. Then we can decompose $P_2(I_n^{2k})=\cup_{j=0}^2P_{2,j}(I_n^{2k}) $ where\begin{align*} P_{2,0}(I_n^{2k})&=\{({R_1},...,  {R_{2k}})\in P_2(I_n^{2k})|A_0=\emptyset\},\\P_{2,1}(I_n^{2k})&=\{({R_1},...,  {R_{2k}})\in P_2(I_n^{2k})|A_0\neq\emptyset,\ \Psi_{R_1}\cdots \Psi_{R_k}=\pm I\},\\P_{2,2}(I_n^{2k})&=\{({R_1},...,  {R_{2k}})\in P_2(I_n^{2k})|\Psi_{R_1}\cdots \Psi_{R_k}\neq\pm I\}. \end{align*}If $({R_1},...,  {R_{2k}})\in P_{2,0}(I_n^{2k}) $, then $J_{R_1}\cdots J_{R_k}$ and $J_{R_{k+1}}\cdots J_{R_{2k}}$ are independent, hence,  $\text{cov}(J_{R_1}\cdots J_{R_k},J_{R_{k+1}}\cdots J_{R_{2k}})=0.$ If $({R_1},...,  {R_{2k}})\in P_{2,2}(I_n^{2k}) $ then\\ $\Tr \Psi_{R_1}\cdots \Psi_{R_k}=0 $ as we discussed in \S \ref{basicfermion}.  Therefore,  by \eqref{tr} again,  we have\begin{align*} |V_1|&=\frac 1{L_n^2}  \frac {1}{{n\choose q_n}^{k}} \left|\sum_{({R_1},...,  {R_{2k}})\in P_{2,1}(I_n^{2k})} \text{cov}(J_{R_1}\cdots J_{R_k},J_{R_{k+1}}\cdots J_{R_{2k}})\cdot\right.\\  &\left.\Tr \Psi_{R_1}\cdots \Psi_{R_k}\Tr \Psi_{R_{k+1}}\cdots \Psi_{R_{2k}}\right|\\ &\leq \frac {c_{2k}}{{n\choose q_n}^{k}} \sum_{({R_1},...,  {R_{2k}})\in P_{2,1}(I_n^{2k})}1= c_{2k}{n\choose q_n}^{-k}{|P_{2,1}(I_n^{2k})|}. \end{align*}
Now we estimate $ |P_{2,1}(I_n^{2k})|$. Let $m=|A_0|>0,$ then there exists $1\leq i_1<\cdots<i_m\leq k$ and $k+1\leq i_1'<\cdots<i_m'\leq 2k$ such that $A_0=\{R_{i_1},\cdots,R_{i_m}\}=\{R_{i_1'},\cdots,R_{i_m'}\}.$ For every $A\in A_1\setminus A_0$, $A$ appears exactly twice in $({R_1},...,  {R_{k}})$ and $ k-m=2|A_1\setminus A_0|$ is even. Similarly $ k-m=2|A_2\setminus A_0|.$ Moreover, we have $\Psi_{R_1}\cdots \Psi_{R_k}=\pm \Psi_{R_{i_1}}\cdots \Psi_{R_{i_m}} $ and \begin{align*} P_{2,1}(I_n^{2k})&=\{({R_1},...,  {R_{2k}})\in P_2(I_n^{2k})|m>0,\ \Psi_{R_{i_1}}\cdots \Psi_{R_{i_m}}=\pm I\}. \end{align*}For every fixed $A_0$, there are $ {|I_n|-m\choose k-m}{k-m\choose \frac{k-m}{2}}$ choices of $ (A_1\setminus A_0,A_2\setminus A_0)$. For every fixed $A_0,\ A_1\setminus A_0$ and $A_2\setminus A_0,$ there are $ (k!/2^{\frac{k-m}{2}})^2=(k!)^2/2^{{k-m}}$ choices of $({R_1},...,  {R_{2k}}). $ Let's denote \begin{equation}\label{bm}B_m=\{({R_1},...,  {R_{m}})\in I_n^{m}|\Psi_{R_1}\cdots \Psi_{R_m}=\pm I,\ R_i\neq R_j,\ \forall\ 1\leq i<j\leq m\}.\end{equation}  Then for fixed $m$, every $A_0$ corresponds exactly $m!$ elements in $B_m,$ thus the number of elements in $P_{2,1}(I_n^{2k}) $ satisfying $|A_0|=m$ is ${|I_n|-m\choose k-m}{k-m\choose \frac{k-m}{2}}(k!)^2/(2^{{k-m}}m!)\cdot|B_m| $ (or $0$ if $k-m$ is odd). It remains to estimate $|B_m|.$

 We first have $B_1=B_2=\emptyset$ by definition. For $m\geq 3$, if $({R_1},...,{R_{m-1}},  {R_{m}})$ and $({R_1},...,{R_{m-1}},  {R_{m}'})$ are both in $B_m $, then $\Psi_{R_m}=\pm \Psi_{R_m'} $, then we must have ${R_m}={R_m'}$ by the discussion in \S \ref{basicfermion}. Therefore,  every $({R_1},...,{R_{m}})\in B_m $ is uniquely determined by its first $m-1$ components $({R_1},...,{R_{m-1}})$ and hence  $ |B_m|\leq |I_n|^{m-1}.$ Now we have \begin{align*} |P_{2,1}(I_n^{2k})|&=\sum_{0<m\leq k,2|k-m}{|I_n|-m\choose k-m}{k-m\choose \frac{k-m}{2}}(k!)^2/(2^{{k-m}}m!)\cdot|B_m|\\&\leq\sum_{0<m\leq k,2|k-m}c_{k,m}|I_n|^{k-m}\cdot|B_m|\\&\leq\sum_{0<m\leq k,2|k-m}c_{k,m}|I_n|^{k-m}|I_n|^{m-1}\\&=c_k|I_n|^{k-1}. \end{align*} Therefore, using $|I_n|={n\choose q_n}$, we finally have \begin{align*} \text{var}[L_n^{-1} \Tr H^k]&=V_1+V_2\leq c_{2k}{n\choose q_n}^{-k}{|P_{2,1}(I_n^{2k})|} +c_k{n\choose q_n}^{-1}\\  &\leq c_k{n\choose q_n}^{-1}. \end{align*}This completes the proof of Lemma \ref{lemma6}.\end{proof}
Now by Lemma \ref{lemma6},  for even $2\leq q_n<n$, we have
\begin{align*} &\e \sum_n|L_n^{-1} \Tr H^k-L_n^{-1}\e[\Tr H^k]|^2\\&=\sum_n\text{var}[L_n^{-1} \Tr H^k]\\ &\leq \sum_n c_k{n\choose q_n}^{-1}\leq \sum_n c_k{n\choose 2}^{-1}<+\infty. \end{align*} Therefore, we have
 $$\sum_n|L_n^{-1} \Tr H^k-L_n^{-1}\e[\Tr H^k]|^2<+\infty,\ a.s. $$ and hence, $$\lim\limits_{n\to+\infty}|L_n^{-1} \Tr H^k-L_n^{-1}\e[\Tr H^k]|=0,\ a.s. $$ Since $L_n^{-1} \Tr H^k=\langle x^k, \rho_n\rangle $ and $\lim\limits_{n\to+\infty}L_n^{-1}\e[\Tr H^k]=\langle x^k, \rho_{\infty}\rangle, $   we have $$\lim\limits_{n\to+\infty}\langle x^k,  \rho_n \rangle=\langle x^k, \rho_{\infty} \rangle,\ a.s. $$ The sequence  $\rho_n(\lambda) $ is tight almost surely and every limiting distribution of its subsequence has the same $k$-th moment as $\rho_{\infty}(\lambda)$ almost surely. Since $\rho_{\infty}(\lambda) $ satisfies Carleman's condition,  every limiting distribution must be $\rho_{\infty}(\lambda)$ almost surely  and this implies $$\lim\limits_{n\to+\infty}   \rho_n(\lambda)= \rho_{\infty}(\lambda),\ a.s. $$ in the sense of distribution. Therefore, we finish the proof of Theorem \ref{main1}.

\section{Proof of Theorem \ref{main2}}\label{theorem2proof}
In this section, we will sketch the proof of Theorem \ref{main2} for the case when $q_n$ is odd. The proof is similar to that of the even case in Theorem \ref{main1} .

For $q_n$ odd, we  consider the Hermitian matrices $$H=i^{(q_n-1)/2}\frac{1}{\sqrt{{{n} \choose {q_n}}}}\sum_{1\leq i_1<i_2< \cdots <i_{q_n} \leq n} J_{i_1i_2\cdots i_{q_n}}{\p}_{i_1}{\p}_{i_2}\cdots {\p}_{i_{q_n}}.$$
Let's first prove Theorem \ref{main2} for the simplest case when $q_n=1$. For such case, we have $$ H^2=\frac{1}{n}\sum\limits_{j=1}^nJ_j^2,\,\,\,\,\Tr H=0,$$ therefore, we have $$ \rho_n(\lambda)=\frac{1}{2}(\delta_{\sqrt{a_n}}(\lambda)+\delta_{-\sqrt{a_n}}(\lambda)) $$ with $$a_n=\frac{1}{n}\sum\limits_{j=1}^nJ_j^2.$$   By the law of large numbers, we have $a_n\to 1\ a.s.$,  this implies $$ \rho_n\to\frac{1}{2}(\delta_{1}+\delta_{-1}). $$ 
For the odd case with $1\leq q_n\leq n/2,$    Lemma \ref{lemma6} is still true (with $q_nk/2 $ replaced by $(q_n-1)k/2 $ in the proof).  Therefore, we only need to prove
  $$\lim\limits_{n\to+\infty}\langle x^k, \e \rho_n \rangle=\langle x^k, \rho_{\infty} \rangle. $$
As before, we have \begin{equation}\label{mnk2}m_{n,k}   = \frac 1{L_n}   \e[ \Tr H^k]= \frac {i^{(q_n-1)k/2}}{L_n{n\choose q_n}^{k/2}} \sum_{{R_1},...,  {R_k}\in I_n} \e[J_{R_1}\cdots J_{R_k}]  \Tr \Psi_{R_1}\cdots \Psi_{R_k}. \end{equation}We still  divide the above summation  as\begin{equation}\label{deco1}m_{n,k}=\sum_{({R_1},...,  {R_k})\in  P_2(I_n^k)}+\sum_{({R_1},...,  {R_k})\in I^k_n\setminus P_2(I_n^k)}.  \end{equation}If $k$ is odd, then we have $P_2(I_n^k)=\emptyset$, and thus  the second inequality in Lemma \ref{lemma1} together with \eqref{tr} will imply that
$$|m_{n,k}|\leq c_k{n\choose q_n}^{-1/2} \to 0.$$
All the rest is to estimate the moment $m_{n,k}$ when $k$ is even.
Similarly, Lemma \ref{lemma1} also implies that the second summation in   \eqref{deco1}  satisfies
$$| \sum_{({R_1},...,  {R_k})\in I^k_n\setminus P_2(I_n^k)}|\to 0.$$
First, for the case $q_n>1$ and $\ q_n^2/n\to a\in[0,+\infty)$, we can rewrite
\begin{align*}\sum _ {P_2(I_n^k)}= \frac 1{L_n} \frac {i^{(q_n-1)k/2}}{{n\choose q_n}^{k/2}} \sum_{\pi\in S_k} \sum_{R_1,\cdots, R_{\frac k2}\in I_n, R_i \neq R_j \,\mbox{if}\,\, i\neq j}   \frac{\Tr \Psi_{R_{\pi(1)}}\cdots \Psi_{R_{\pi(k)}}}{(k/2)!}.
\end{align*}Now Lemma \ref{lemma3} is replaced by $$\frac {i^{(q_n-1)k/2}} {L_n} \Tr  \Psi_{R_{\pi(1)}}\cdots \Psi_{R_{\pi(k)}}=(-1)^{\sum_{k=1}^{\kappa(\pi)}(|R_{r_k}\cap R_{s_k}|+1)}.$$ If we combine this with Lemma \ref{lemma4}, for any fixed  map $\pi$ and odd $q_n$, we have
\begin{align*}&\lim_{n\to \infty}  \frac 1{L_n} \frac {i^{(q_n-1)k/2}}{{n\choose q_n}^{k/2}} \sum_{R_1,\cdots, R_{\frac k2}\in I_n, R_i \neq R_j \,\mbox{if}\,\, i\neq j}   \Tr \Psi_{R_{\pi(1)}}\cdots \Psi_{R_{\pi(k)}}=(-e^{-2a})^{\kappa(\pi)}.
\end{align*}
Therefore, we have
\begin{equation}\label{malk}\lim_{n\to \infty}\sum _ {P_2(I_n^k)}
 = \frac{1}{(k/2)!}\sum_{\pi\in S_k} (-e^{-2a})^{\kappa(\pi)}:=\tilde m_k^a.
\end{equation}
Let $\tilde m_k^a=0 $ for $k$ odd, then we have proved
$$\lim_{n\to\infty} m_{n,k}=
   \tilde m^a_k. $$For the case $a=0$, we can check directly that $\tilde m^a_k=1$ for $k=2,4$ (in fact for all $k$ even). Therefore, we have $\langle x,\rho_{\infty}\rangle=0,\ \langle x^2,\rho_{\infty}\rangle=\langle x^4,\rho_{\infty}\rangle=1$ and $ \langle (x^2-1)^2,\rho_{\infty}\rangle=0$, this implies $\rho_\infty=\frac{1}{2}(\delta_{1}+\delta_{-1})$. For the case $a>0$, we can use the theory of $q$-Hermite polynomials as in the proof of case 2 in Theorem \ref{main1} to conclude the result. 
   
   In the end, for the case $q_n^2/n\to \infty$ and $q_n\leq n/2,$ we can use Lemma \ref{lemma5} to conclude the result, and thus we finish the proof of Theorem \ref{main2}.

\section{Largest eigenvalue}\label{largg}

Now let's prove Theorem \ref{2case}  in  detail.  Although many facts are known to physicists  when $q_n=2$ \cite{black, GV3, MS}, we still prove \eqref{expressions} \eqref{ddsds} to make the article self-contained.
 


 Let $B=(J_{ij})_{1\leq i,j\leq n},\ J_{ji}=-J_{ij} $ be the real antisymmetric matrix. We assume the eigenvalues of $B$ are $\pm i\mu_j$ where  $\mu_j\geq 0$ for $ 1\leq j\leq n/2$. Then there exists an orthogonal matrix $A=(a_{ij})_{1\leq i,j\leq n}\in O(n)$ such that $ A^TBA=C=(c_{ij})_{1\leq i,j\leq n}$ where $c_{2j-1,2j}=\mu_j, \ c_{2j,2j-1}=-\mu_j$ and other values of $c_{ij}$ are 0. Now we have $B=ACA^T$ and $J_{ij}=\sum\limits_{1\leq k,l\leq n}a_{ik}c_{kl}a_{jl}.$ Therefore, we can rewrite \begin{align*}H&=\frac{i}{\sqrt{{{n} \choose {2}}}}\sum_{1\leq i_1<i_2\leq n} J_{i_1i_2}{\p}_{i_1}{\p}_{i_2}=\frac{i}{2\sqrt{{{n} \choose {2}}}}\sum_{1\leq i_1,i_2\leq n} J_{i_1i_2}{\p}_{i_1}{\p}_{i_2}\\&=\frac{i}{2\sqrt{{{n} \choose {2}}}}\sum_{1\leq i_1,i_2,k,l\leq n} a_{i_1k}c_{kl}a_{i_2l}{\p}_{i_1}{\p}_{i_2}\\&=\frac{i}{2\sqrt{{{n} \choose {2}}}}\sum_{1\leq k,l\leq n} c_{kl}\left(\sum_{i_1=1}^{n}a_{i_1k}{\p}_{i_1}\right)\left(\sum_{i_2=1}^{n}a_{i_2l}{\p}_{i_2}\right)\\&=\frac{i}{2\sqrt{{{n} \choose {2}}}}\sum_{1\leq k,l\leq n} c_{kl}\widetilde{\p}_{k}\widetilde{\p}_{l}=\frac{i}{\sqrt{{{n} \choose {2}}}}\sum_{j=1}^{n/2} \mu_{j}\widetilde{\p}_{2j-1}\widetilde{\p}_{2j},\end{align*}where $\widetilde{\p}_{k}=\sum\limits_{j=1}^{n}a_{jk}{\p}_{j}.$ Since $A$ is orthogonal,  we have\begin{align*}\{\widetilde{\p}_{k},\widetilde{\p}_{l}\}&=\sum_{1\leq i_1,i_2\leq n}a_{i_1k}a_{i_2l}\{\widetilde{\p}_{i_1},\widetilde{\p}_{i_2}\}=\sum_{1\leq i_1,i_2\leq n}a_{i_1k}a_{i_2l}(2\delta_{i_1i_2})\\&=\sum_{1\leq i_1\leq n}2a_{i_1k}a_{i_1l}=2\delta_{kl}.\end{align*}Thus,  $\widetilde{\p}_j$ are also Majorana fermions. Since ${\p}_j$ is Hermitian for all $1\leq j\leq n$, so is $\widetilde{\p}_j$. Let $A_j=i\widetilde{\p}_{2j-1}\widetilde{\p}_{2j},$ then $A_j$ is also Hermitian. Furthermore, since $A_j^2=I$ and $A_jA_k=A_kA_j,$ the eigenvalues of $A_j$ are $\pm1$, and the eigenvalues of $H$ are in the form of \begin{equation}\label{expressions}{{n} \choose {2}}^{-\frac{1}{2}}\sum\limits_{j=1}^{n/2}\pm \mu_{j}.\end{equation}In particular,  $\lambda_{max}\leq {{n} \choose {2}}^{-\frac{1}{2}}\sum\limits_{j=1}^{n/2} \mu_{j}.$ Now we show that the equality holds.

 By $A_j^2=I$ and  $A_jA_k=A_kA_j$ again, $A_j\ (1\leq j\leq n/2)$ have a common eigenvector $0\neq e_0\in \mathbb{C}^{L_n}$ and $A_je_0=\pm e_0.$  Let   $R=\{2j|A_je_0=- e_0\}$ and define $\widetilde\Psi_R$ as usual, then we have $A_j\widetilde{\Psi}_R=\widetilde{\Psi}_RA_j$ if $A_je_0=e_0$ and $A_j\widetilde{\Psi}_R=-\widetilde{\Psi}_RA_j$ if $A_je_0=-e_0$. Let $e_1=\widetilde{\Psi}_Re_0$, then $e_1\neq 0$ and $A_je_1=e_1$. Thus we have $$ He_1=\left[{{n} \choose {2}}^{-\frac{1}{2}}\sum\limits_{j=1}^{n/2} \mu_{j}\right]e_1,$$ i.e., $ {{n} \choose {2}}^{-\frac{1}{2}}\sum\limits_{j=1}^{n/2} \mu_{j}$ is an eigenvalue of $H$ and hence 
\begin{equation}\label{ddsds}\lambda_{max}= {{n} \choose {2}}^{-\frac{1}{2}}\sum\limits_{j=1}^{n/2} \mu_{j}.\end{equation}
  To study $\lambda_{max}, $ we need the semicircle law of the eigenvalues $\pm \mu_j.$ Let \begin{equation}
	  \rho_n^*(\lambda)=\frac 1{n}\sum_{j=1}^{n/2} (\delta_{\mu_{j}/\sqrt{n-1}}(\la)+\delta_{-\mu_{j}/\sqrt{n-1}}(\la)).	 \end{equation}
	
	  The following lemma regarding the distribution of eigenvalues of random antisymmetric  matrices is standard  \cite{M},
\begin{lem}\label{lemma7}$\rho_n^*\to\frac{1}{2\pi}\sqrt{4-x^2} \mbox{\Large$\chi$}_{[-2,2]}\ a.s.$ in the sense of distribution.
\end{lem} Since \begin{align*}\langle x^2, \rho_n^* \rangle=\frac 2{n}\sum_{j=1}^{n/2} \left|\frac{\mu_{j}}{\sqrt{n-1}}\right|^2=\frac 2{n(n-1)}\sum_{j=1}^{n/2} {\mu_{j}}^2=\frac 2{n(n-1)}\sum_{1\leq i_1<i_2\leq n}J_{i_1i_2} ^2,\end{align*}we have \begin{align*}\e\langle x^2, \rho_n^* \rangle=\frac 2{n(n-1)}\sum_{1\leq i_1<i_2\leq n}\e J^2_{i_1i_2}=\frac 2{n(n-1)}{{n} \choose {2}}=1.\end{align*}Thus $x$ is uniformly integrable with respect to $\rho_n^* \ a.s.$ and $\langle |x|, \rho_n^* \rangle$ is uniformly integrable. Moreover, we have\begin{align*}\lambda_{max}= {{n} \choose {2}}^{-\frac{1}{2}}\sum\limits_{j=1}^{n/2} \mu_{j}= {{n} \choose {2}}^{-\frac{1}{2}}\frac{n}{2}\sqrt{n-1}\langle |x|, \rho_n^* \rangle= \sqrt{\frac{n}{2}}\langle |x|, \rho_n^* \rangle.\end{align*}By Lemma \ref{lemma7}, we have \begin{align*}\lim_{n\to+\infty}\langle |x|, \rho_n^* \rangle=\frac{1}{2\pi}\int_{-2}^2|x|\sqrt{4-x^2}dx=\frac{8}{3\pi}\ a.s.\ \text{and in}\ L^1.\end{align*}Thus we have \begin{align*}\lim_{n\to+\infty}\frac{\lambda_{max}}{\sqrt{n}}=\frac{4\sqrt{2}}{3\pi}\ a.s.\ \text{and in}\ L^1, \end{align*}
which finishes the proof of Theorem \ref{2case}.

\subsection{Proof of Theorem \ref{gausscase}}
Now let's finish the proof  of Theorem \ref{gausscase}. 
 
\begin{proof}
We first need the following estimate
\begin{equation}\label{claim1}\frac 1{L_n}\e [\Tr {H}^{k}]\leq \e\left(\frac{1}{\sqrt{{{n} \choose {q_n}}}}\sum\limits_{1\leq i_1<i_2< \cdots <i_{q_n} \leq n} J_{i_1i_2\cdots i_{q_n}}\right)^k.\end{equation}
To prove \eqref{claim1}, we first note \begin{equation}\label{mnk2} \frac 1{L_n}   \e[ \Tr H^k]= \frac 1{L_n} \frac {i^{q_nk/2}}{{n\choose q_n}^{k/2}} \sum_{{R_1},...,  {R_k}\in I_n} \e[J_{R_1}\cdots J_{R_k}]  \Tr \Psi_{R_1}\cdots \Psi_{R_k} \end{equation}and\begin{equation}\label{mnk3}\e\left(\frac{1}{\sqrt{{{n} \choose {q}}}}\sum\limits_{1\leq i_1<i_2< \cdots <i_{q_n} \leq n} J_{i_1i_2\cdots i_{q_n}}\right)^k= \frac {1}{{n\choose q_n}^{k/2}} \sum_{{R_1},...,  {R_k}\in I_n} \e[J_{R_1}\cdots J_{R_k}].\end{equation}
Since $\e[J_{R_1}\cdots J_{R_k}]\geq0 $ is always true,  then the estimate \eqref{claim1} follows \eqref{tr}.

For the partition function defined in \eqref{parti}, for any $ \beta>0$, we have the estimate
\begin{align*}\e Z(-\beta)&=\sum_{k=0}^{+\infty}\frac{\beta^{k}}{k!}\e [\Tr {H}^{k}]\\&\leq L_n\sum_{k=0}^{+\infty}\frac{\beta^{k}}{k!}\e\left(\frac{1}{\sqrt{{{n} \choose {q_n}}}}\sum\limits_{1\leq i_1<i_2< \cdots <i_{q_n} \leq n} J_{i_1i_2\cdots i_{q_n}}\right)^k\\&=L_n\e\exp\left(\frac{\beta}{\sqrt{{{n} \choose {q_n}}}}\sum\limits_{1\leq i_1<i_2< \cdots <i_{q_n} \leq n} J_{i_1i_2\cdots i_{q_n}}\right)\\&=L_ne^{\beta^2/2}.\end{align*}Here, we used the fact that $\frac{1}{\sqrt{{{n} \choose {q_n}}}}\sum\limits_{1\leq i_1<i_2< \cdots <i_{q_n} \leq n} J_{i_1i_2\cdots i_{q_n}} $ is the standard Gaussian random variable.

By definition of the partition function again, we further have $$\e Z(-\beta)\geq\e(e^{\beta \la_{max}}).$$ Then by Jensen's inequality, for any $ \beta>0$, we have \begin{align*}\beta\e \la_{max}\leq \ln\e(e^{\beta \la_{max}})\leq\ln \e Z(-\beta)\leq \ln(L_ne^{\beta^2/2})=\ln L_n+\beta^2/2.\end{align*} Choosing $ \beta=\sqrt{2\ln L_n}=\sqrt{2\ln 2^{n/2}}=\sqrt{2\cdot n/2\cdot\ln 2}=\sqrt{n\ln 2}$,  we have$$\e \la_{max}\leq (\ln L_n+\beta^2/2)/\beta=(\beta^2/2+\beta^2/2)/\beta=\beta=\sqrt{n\ln 2},$$
which finishes Theorem \ref{gausscase}. \end{proof}

\end{document}